\documentclass[11pt]{article}

\usepackage[utf8]{inputenc}
\usepackage{geometry}
\usepackage[ruled,vlined,linesnumbered]{algorithm2e}
\usepackage{algpseudocode}
\usepackage{amsmath,amsfonts,amsthm,amssymb,dsfont}
\usepackage{color}
\usepackage{graphicx}

\usepackage{thmtools}

\usepackage{xcolor}
\usepackage{nameref}
\definecolor{ForestGreen}{rgb}{0.1333,0.5451,0.1333}
\definecolor{DarkRed}{rgb}{0.65,0,0}
\definecolor{Red}{rgb}{1,0,0}
\usepackage[linktocpage=true,
pagebackref=true,colorlinks,
linkcolor=DarkRed,citecolor=ForestGreen,
bookmarks,bookmarksopen,bookmarksnumbered]{hyperref}
\usepackage{cleveref}

\newcommand{\tO}{\tilde{O}}
\newcommand{\tOmega}{\tilde{\Omega}}

\declaretheorem{theorem}
\declaretheorem[numberlike=theorem]{lemma}

\crefname{algorithm}{Algorithm}{Algorithms}
\Crefname{algorithm}{Algorithm}{Algorithms}

\theoremstyle{definition}
\declaretheorem[numberlike=theorem]{definition}

\title{Minimum Cuts in Directed Graphs via $\sqrt{n}$ Max-Flows}
\author{ Ruoxu Cen\\Duke University \and Jason Li\\Carnegie Mellon University \and Danupon Nanongkai\\University of Copenhagen \& KTH \and Debmalya Panigrahi\\Duke University \and Thatchaphol Saranurak\\University of Michigan, Ann Arbor}

\begin{document}

\maketitle

\begin{abstract}
We give an algorithm to find a mincut in an $n$-vertex, $m$-edge weighted directed graph using $\tO(\sqrt{n})$ calls to {\em any} maxflow subroutine. Using state of the art maxflow algorithms, this yields a directed mincut algorithm that runs in $\tO(m\sqrt{n} + n^2)$ time. This improves on the 30 year old bound of $\tO(mn)$ obtained by Hao and Orlin for this problem.
\end{abstract}

\pagenumbering{gobble}

\clearpage

\pagenumbering{arabic}

\section{Introduction}
The {\em minimum cut} (or mincut) problem is one of the most widely studied problems in graph algorithms. In directed graphs (or {\em digraphs}), the goal of this problem is to find a minimum weight set of edges whose removal creates multiple strongly connected components in the graph. Equivalently, a minimum cut of a digraph is a bipartition of the vertices into two non-empty sets $(S, V\setminus S)$ such that the weight of edges from $S$ to $V\setminus S$ is minimized. This problem can be solved by using $O(n)$ maxflow calls by finding the minimum among all the $r-v$ mincuts and $v-r$ mincuts for all vertices $v\not= r$. In a beautiful result, Hao and Orlin~\cite{HaoO92} showed that these maxflow calls can be amortized to match the running time of only polylog($n$) calls to the {\em push-relabel maxflow algorithm}~\cite{GoldbergT88}. This leads to an overall running time of $\tO(mn)$ on an $m$-edge, $n$-vertex graph. Since their work, better maxflow algorithms have been designed (e.g., Goldberg and Rao~\cite{GoldbergR98}), but the amortization does not work for these algorithms. Using a different technique of duality between rooted mincuts and arborescences, Gabow~\cite{Gabow1995} obtained a running time of $\tO(m\lambda)$ for this problem, where $\lambda$ is the weight of a mincut (assuming integer weights). This is at least as good as the Hao-Orlin running time for unweighted simple graphs, but can be much worse for weighted graphs. Indeed, the Hao-Orlin bound of $\tO(mn)$ remains the state of the art for the directed mincut problem on arbitrary weighted graphs. 

In this paper, we give an algorithm for the directed mincut problem that has a time complexity of $O(\sqrt{n})$ maxflow calls. Importantly, and unlike the Hao-Orlin algorithm, our algorithm can use {\em any} maxflow algorithm; in fact, it treats the maxflow algorithm as a black box. Using state of the art max flow algorithms that run in $\tO(m+n^{3/2})$ time~\cite{BrandLLSSSW21}, this yields a directed mincut algorithm in $\tO(m\sqrt{n} + n^2)$ time, thereby improving on the Hao-Orlin bound. %
Moreover, if one were to believe the widely held conjecture that maxflow would eventually be solved in $\tO(m)$ time, then the running time of our algorithm would automatically become $\tO(m\sqrt{n})$.

\begin{theorem}
\label{thm:main}
There is a randomized Monte Carlo algorithm that finds a minimum cut whp in $\tO(m\sqrt{n} + n^2)$ time on an $n$-vertex, $m$-edge directed graph.
\end{theorem}

\paragraph{Our Techniques.}
At a high level, our paper is inspired by Karger's celebrated near-linear time mincut algorithm in undirected graphs~\cite{Karger00}. Karger's algorithm has three main steps: (a) sparsify the graph by random sampling of edges to reduce the mincut value to $O(\log n)$, (b) use a semi-duality between mincuts and spanning trees to pack $O(\log n)$ edge-disjoint spanning trees in the sparsifier, and (c) find the minimum weight cut among those that have only one or two edges in each such spanning tree using a dynamic program. But, directed graphs are substantially different from undirected graphs. In particular, steps (a) and (c) are not valid in a directed graph. We cannot hope to sparsify a directed graph since many directed graphs do not have sparsifiers even in an existential sense. Moreover, even if a mincut had just a single edge in a spanning tree, Karger's dynamic program to recover this cut cannot be used in a directed graph. 

To overcome these challenges, we adopt several ingredients that we outline below:
\begin{itemize}
    \item We consider two possibilities: either the mincut has $\tOmega(\sqrt{n})$ vertices on the smaller side or fewer (let us call these {\em balanced} and {\em unbalanced} cuts respectively). If the mincut is a balanced cut, we use two random samples of $\tO(\sqrt{n})$ and $\tO(1)$ vertices each, and find $s-t$ mincuts for all pairs of vertices from the two samples. It is easy to see that whp, the two samples would respectively {\em hit} the smaller and larger sides of the mincut, and hence, one of these $s-t$ mincuts will reveal the overall mincut of the graph.
    \item The main task, then, is to find the mincut when it is unbalanced. In this case, we use a sequence of steps. The first step is to use {\em cut sparsification} of the graph by random sampling of edges. This scales down the size of the mincut, but unlike in an undirected graph, all the cuts of a digraph do not necessarily converge to their expected values in the sample. However, crucially, {\em the mincut can be scaled to $\tO(\sqrt{n})$ while ensuring that all the unbalanced cuts converge to their expected values.}
    \item Since only the unbalanced cuts converge to their expected values, it is possible that some balanced cut is the new mincut of the sampled graph, having been scaled down disproportionately by the random sampling. Our next step is to {\em overlay} this sampled graph with an {\em expander} graph, a technique inspired by recent work of Li~\cite{Li21}.  Note that an expander has a larger weight for balanced cuts than for unbalanced cuts. We choose the expansion of the graph carefully so that the balanced cuts get sufficiently large weight of edges that they are no longer candidates for the mincut of the sample, while the unbalanced cuts are only distorted by a small multiplicative factor. 
    \item At this point, we have obtained a graph where the original mincut (which was unbalanced) is a near-mincut of the new graph. Next, we create a (fractional) packing of edge-disjoint arborescences\footnote{An {\em arborescence} is a spanning tree in a directed graph where all the edges are directed away from the root.} in this graph using a multiplicative weights update procedure (e.g.,~\cite{Young1995}). By duality, these arborescenes have the following property: if we sample $O(\log n)$ random arborescences from this packing, then there will be at least one arborescence whp such that the original mincut $1$-respects the arborescence. (A cut $1$-respects an arborescence if the latter contains just one edge from the cut.) 
    \item Thus, our task reduces to the following: given an arborescence, find the minimum weight cut in the original graph among all those that $1$-respect the arborescence. Our final technical contribution is to give an algorithm that solves this problem using $O(\log n)$ maxflow computations. For this purpose, we use a centroid-based recursive decomposition of the arborescence, where in each step, we use a set of maxflow calls that can be amortized on the original graph. The minimum cut returned by all these maxflow calls is eventually returned as the mincut of the graph.
\end{itemize}

We note that unlike both the Hao-Orlin algorithm and Gabow's algorithm that are both deterministic algorithms, our algorithm is randomized (Monte Carlo) and might yield the wrong answer with a small (inverse polynomial) probability. Derandomizing our algorithm, or matching our running time bound using a different deterministic algorithm, remains an interesting open problem.

\paragraph{Previous Work.}
The mincut problem has been studied in directed graphs over several decades. For unweighted graphs, Even and Tarjan~\cite{EvenT75} gave an algorithm for this problem that runs in $O(mn \cdot\min(\sqrt{m}, n^{2/3}))$ time. This was improved by Schnorr~\cite{Schnorr79} improved by this bound for certain graphs to $O(mn\lambda)$, where $\lambda$ is the value of the directed mincut. This was further improved by Mansour and Schieber~\cite{MansourS89} to $O(n\cdot \min(m, \lambda^2 n))$ after almost a decade of Schnorr's work. Mansour and Schieber's bound of $O(mn)$ was matched up to logarithmic factors for the more general case of weighted digraphs by Hao and Orlin~\cite{HaoO92}. Finally, Gabow~\cite{Gabow1995} gave an algorithm that runs in $\tO(m\lambda)$ which further refines this bound for graphs with small $\lambda$. These remained the fastest directed mincut algorithms for almost 30 years before our work.

\paragraph{Concurrent Work.}
Two recent results on algorithms for finding mincuts in directed graphs were obtained concurrently and independently of our work. First, Chekuri and Quanrud \cite{ChekuriQ21}  showed an exact algorithm with running time $\tO(n^2 U)$ if edge weights are integers between $1$ and $U$.\footnote{We note that we can use the degree reduction technique in \cite{ChekuriQ21} to speed up our algorithm from time $\tO(m\sqrt{n}+n^2)$ to $\tO(n^2)$, but we omit details of this improvement in this manuscript to avoid interdependence on concurrent, unpublished research.} 
Second, Quanrud \cite{Quanrud21} has obtained an $(1+\epsilon)$-approximate algorithm that runs in $\tilde{O}(n^2/\epsilon^2)$ time.\footnote{Quanrud can also obtain $o(mn)$-time algorithms using the currently fastest maxflow algorithm on sparse graphs by Gao, Liu, and Peng \cite{GaoLP21}. We can also obtain the same time but for exact mincuts.} 
Both papers also extend their ideas to obtain approximation results for other problems as well, such as the vertex mincut problem.

\section{The Directed Min-Cut Algorithm}

Given a directed graph $G = (V, E)$ with non-negative edge weights $w$, we consider the problem of finding a (global) minimum directed cut in this graph.
For simplicity, we assume that all edge weights $w$ are integers and are polynomially bounded.
We denote $\overline{U}=V\setminus U$.
Let $\partial U$ be the set of edges in the cut $(U, \overline{U})$, and
let $\delta(U)$ be the weight of the cut, i.e.,  $\delta(U)=\sum_{e\in \partial U}w(e)$.
Our goal is to find $\arg\min_{\emptyset\subset U \subset V} \delta(U)$.
Let $MF(m,n)$ denote the time complexity of $s$-$t$ maximum flow on a digraph with $n$ vertices and $m$ edges. The current record for this bound is $MF(m, n) = \tO(m + n^{3/2})$~\cite{BrandLLSSSW21}. We emphasize that our directed mincut algorithm uses maxflow subroutines in a black box manner and therefore, any maxflow algorithm suffices. Correspondingly, we express our running times in terms of $MF(m, n)$.

Next we describe the algorithm.
Let $S^*$ be the source side of a minimum cut. The algorithm considers the following two cases, computes a cut for each case and takes the smaller of the two cuts as its final output.
 \begin{enumerate}
     \item The first case aims to compute the correct mincut in the event that $\min\{|S^*|,|\overline{S^*}|\} > \theta\cdot \sqrt{n}/\log n$. In this case, we randomly sample two vertices $s,t\in V$, then with reasonable probability, they will lie on opposite sides of the mincut. In that case, we can simply compute the maxflow from $s$ to $t$. Repeating the sampling $O(\sqrt{n}\log^2 n)$ times, we obtain the mincut whp. The total running time for this case is $O(MF(m,n)\sqrt{n}\log^2 n)$ and is formalized in Lemma \ref{lem:size-balanced} below:
     \begin{lemma}
\label{lem:size-balanced}
If $\min\{|S^*|,|\overline{S^*}|\} > r$, then whp a mincut can be calculated in time $O(MF(m,n)\cdot (n/r) \cdot \log n)$. 
\end{lemma}
\begin{proof}
Uniformly sample a list of $k = d \cdot (n/r) \cdot \lg n$ vertices $u_1,\ldots,u_k$, where $d$ is a large constant. Wlog, assume $|S^*| \le |\overline{S^*}|$, and let $\eta=\frac{|S^*|}{n}>\frac{r}{n}$.
With probability at least $1-2(1-\eta)^k\ge 1-2e^{-k\eta}\ge 1-2n^{-d}$, the list $u_1,\ldots,u_k$ contains at least one vertex from each of $S^*$ and $\overline{S^*}$. Hence, there exists $i$ such that $u_i$ and $u_{i+1}$ are on different sides of the $(S^*, \overline{S^*})$ cut.
By calculating maxflows for all ($u_i$, $u_{i+1}$) and ($u_{i+1}$, $u_i$) pairs, and reporting the smallest $s-t$ mincut in these calls, we return a global mincut whp.
\end{proof}
     
     \item The second case takes care of the event that $\min\{|S^*|,|\overline{S^*}|\} \le \theta\cdot \sqrt{n}/\log n$. In this case, we select an arbitrary vertex $s$, and give an algorithm for finding an $s$-mincut defined as:
    \begin{definition}
        An \emph{$s$-mincut} is a minimum weight cut among all those that have $s$ on the source side of the cut, i.e., $\arg\min_{\{s\}\subseteq S \subset V} \delta(S)$.
    \end{definition}
    Repeating this process with all edge directions reversed, and returning the smaller of the $s$-mincuts in the original and the reversed graphs, yields the overall mincut. 
     
    We now describe the $s$-mincut algorithm, where we overload notation to denote the value of the $s$-mincut by $\lambda$. Here, we first guess $O(\log n)$ potential values of $\tilde\lambda$, which is our estimate of $\lambda$, as the powers of $2$, one of which lies in the range $[\lambda,2\lambda]$, and then for each $\tilde\lambda$, sparsifies the graph using \Cref{thm:sparsification} from \Cref{sec:sparsification}. For each such sparsifier $H$, the algorithm then applies \Cref{thm:tree-packing} from \Cref{sec:arborescence} to pack $O(\log n)$ $s$-arborescences in $H$ in $O(m\sqrt n\log n)$ time, one of which will $1$-respect the $s$-mincut in $G$ (for the correct value of $\tilde\lambda$):
    \begin{definition}
        An \emph{$s$-arborescence} is a directed spanning tree rooted at $s$ such that all edges are directed away from $s$.
        A directed $s$-cut \emph{$k$-respects} an $s$-arborescence if there are at most $k$ cut edges in the arborescence.
    \end{definition}
    Finally, for each of the $O(\log n)$ $s$-arborescences, the algorithm computes the minimum $s$-cut that $1$-respects each arborescence; this algorithm is described in Algorithm~\ref{alg:tree} and proved in \Cref{thm:tree-alg} from \Cref{sec:maxflows}. It runs in $O((MF(m,n) + m)\cdot \log n)$ time for each of the $O(\log n)$ arborescences.
 \end{enumerate}
Combining both cases, the total running time becomes $\tilde O(m\sqrt n + MF(m,n)\sqrt n)$, which establishes \Cref{thm:main}.

\section{Sparsification}
\label{sec:sparsification}

This section aims to reduce mincut value to $\tO(\sqrt{n})$ while keeping $S^*$ a $(1+\epsilon)$-approximate mincut for a constant $\epsilon > 0$ that we will fix later.
Our algorithm in this stage has two steps. First, we use random sampling to scale down the expected value of all cuts such that the expected value of the mincut $\delta(S^*)$ becomes $\tO(\sqrt{n})$. We also claim that $\partial S^*$ remains an approximate mincut {\em among all unbalanced cuts} by using standard concentration inequalities. However, since the number of balanced cuts far exceeds that of unbalanced cuts, it might be the case that some balanced cut has now become much smaller in weight than all the unbalanced cuts. This would violate the requirement that $\partial S^*$ should be an approximate mincut in this new graph. This is where we need our second step, where we overlay an expander on the sampled graph to raise the values of all balanced cuts above the expected value of $\partial S^*$ while only increasing the value of $\partial S^*$ by a small factor. This last technique is inspired by recent work of Li~\cite{Li21} for a deterministic mincut algorithm in undirected graphs.

We start with the specific expansion properties that we need, and a pointer to an existing construction of such an expander.
\begin{definition}[Expander]
A $\psi$-expander is an undirected graph $G(V,E)$ such that for any $S\subset V$, $\delta_G(S) \ge \psi\cdot \min\{|S|,|\bar{S}|\}$.
\end{definition}
\begin{lemma}[Theorem 2.4 of \cite{Chuzhoy2019}]
\label{lem:expander}
Given integer $n$, we can construct in linear time an $\alpha_0$-expander $X$ for some constant $\alpha_0>0$, such that every vertex in $X$ has degree at most 9.
\end{lemma}

Now, we prove the main property of this section:

\begin{lemma}
\label{thm:sparsification}
Given a digraph $G$, a parameter $\tilde\lambda\in[\lambda,2\lambda]$, and a constant $\epsilon\in(0,1)$,
we can construct in $O(m\log n)$ time a value $p\in(0,1]$ and a digraph $H$ with $O(m)$ edges such that the following holds whp for the value $p=\min\{\frac{\sqrt n}\lambda,1\}$.

\begin{enumerate}
    \item There is a constant $\theta>0$ (depending on $\epsilon$) such that for any set $\emptyset\ne S\subsetneq V$ with $\min\{|S|,|\bar{S}|\} \le \theta\cdot \sqrt{n}/\log n$, we have $$(1-\epsilon)\cdot p\cdot \delta_G(S)\le \delta_H(S)\le(1+\epsilon)\cdot p\cdot \delta_G(S);$$
    \item For any set $\emptyset\ne S\subsetneq V$, $\delta_H(S)\ge (1-\epsilon)p\lambda$.
\end{enumerate}
\end{lemma}

\begin{proof}
If $\tilde\lambda\le 2\sqrt{n}$, then $\lambda\le\tilde\lambda\le2\sqrt n$ as well, so we set $H$ to be $G$ itself, which satisfies all the properties for $p=1$.
For the rest of the proof, we assume that $\tilde\lambda>2\sqrt{n}$, so that $\lambda\ge\sqrt n$, and we set $p=\frac{\sqrt{n}}{\tilde\lambda}\le1$. Throughout the proof, define $\epsilon'=\epsilon/2$,  $r=\frac{\epsilon'^2}{6}\sqrt{n}/\log n$,  $\alpha=\frac{\sqrt{n}}{\alpha_0r}$, and $\theta=\frac{\epsilon'^3\alpha_0}{54}$, where $\alpha_0$ is the constant from Lemma \ref{lem:expander}.

We first construct digraph $\hat{G}$ by reweighting the edges of $G$ as follows. For each edge $e$ in $G$, assign it a random new weight $w_{\hat G}(e)$ chosen according to binomial distribution $B(w(e), p)$. (If $w_{\hat G}(e)=0$, then remove $e$ from $\hat G$.) 
For each set $\emptyset\ne S\subsetneq V$ with $\min\{|S|, |\bar{S}|\} \le r$, we have $\mathds{E}\delta_{\hat{G}}(S)=p\delta_G(S)$, and by Chernoff bound, the probability that $\delta_{\hat{G}}(S)$ falls outside $[(1-\epsilon')p\delta_G(S), (1+\epsilon')p\delta_G(S)]$ is upper-bounded by $2e^{-\lambda\epsilon'^2 /3} \le 2n^{-2r}$.
There are $O(n^r)$ sets $S$ with $\min\{|S|,|\bar{S}|\}\le r$, so by a union bound, whp all such sets satisfy $(1-\epsilon')p\delta_G(S) \le \delta_{\hat{G}}(S) \le (1+\epsilon')p\delta_G(S)$.

Construct graph $X$ according to Lemma \ref{lem:expander} and split each undirected edge into two directed edges.
Let $H$ be the ``union" of $\hat{G}$ and $\alpha X$, so that each edge $e$ in $H$ has weight $w_H(e)=w_{\hat G}(e)+\alpha w_X(e)$, where we say $w(e)=0$ if $e$ does not exist in the corresponding graph.

We now show that $H$ satisfies the two desired properties.
 \begin{enumerate}
 \item For any set $\emptyset\ne S\subsetneq V$ with $\min\{|S|,|\bar{S}|\} \le \theta\cdot \sqrt{n}/\log n= \frac{\epsilon'\alpha_0}{9}r\le r$, we have $\delta_H(S)\ge\delta_{\hat{G}}(S)\ge (1-\epsilon')p\delta_G(S)$ from before, so $\delta_H(S)\ge(1-\epsilon)p\delta_G(S)$ as well. For the upper bound, we have
\[ \delta_H(S)=\delta_{\hat{G}}(S)+\alpha\delta_X(S)\le (1+\epsilon')p\delta_G(S)+9\alpha|S| \le (1+\epsilon')p\delta_G(S)+\epsilon'\sqrt{n} \le (1+\epsilon)p\delta_G(S) .\]
 \item For any set $\emptyset\ne S\subsetneq V$ with $\min\{|S|,|\bar{S}|\} \le \theta\cdot \sqrt{n}/\log n= \frac{\epsilon'\alpha_0}{9}r\le r$, we have  $\delta_H(S)\ge \delta_{\hat{G}}(S)\ge (1-\epsilon) p \delta_G(S) \ge (1-\epsilon)p\lambda$ as required by  property~(2).
When $\min\{|S|,|\bar{S}|\}>r$, we have $\delta_H(S)\ge\alpha\delta_X(S)\ge\alpha\alpha_0 r\ge\sqrt{n}$ for all $\emptyset\ne S\subsetneq V$.

 \end{enumerate}

Finally, $H$ has $O(m)$ edges because $E(\hat{G})$ is a subset of $E(G)$ and $E(X)=O(n)$.
\end{proof}

\section{Finding a 1-respecting Arborescence}
\label{sec:arborescence}

In this section, we assume that there is an unbalanced mincut and show how to obtain an $s$-arborescence that 1-respects the mincut. 
More formally, we prove the following: 

\begin{lemma}
\label{thm:tree-packing}
Given weighted digraph $G$ and a fixed vertex $s$ such that $s$ is in the source side of a minimum cut $S^*$ and $\min\{|S^*|,|\overline{S^*}|\}\le \theta\cdot \sqrt{n}/\log n$ where $\theta$ is defined in \Cref{thm:sparsification}, in $O(m\sqrt{n}\log n)$ time we can find $O(\log n)$ $s$-arborescences, such that whp a minimum cut 1-respects one of them.
\end{lemma}

The idea of this lemma is as follows. First, we apply \Cref{thm:sparsification} to our graph
$G$ and obtain graph $H$. Whp, a mincut $S^{*}$ in $G$ corresponds
to a cut in $H$ of size $(1\pm\epsilon)p\lambda$ and no cut in $H$
has size less than $(1-\epsilon)p\lambda$. That is, $S^{*}$ is a
$(1+O(\epsilon))$-approximate mincut in $H$. It remains to find
an arborescence in $H$ that 1-respects $S^{*}$. To do this, we employ
a multiplicative weight update (MWU) framework. The algorithm begins by setting all
edge weights to be uniform (say, weight $1$). Then, we repeat for $O(\sqrt{n}\log(n)/\epsilon^{2})$
rounds. For each round, we find in near-linear time a minimum weight arborescence and
multiplicatively increase the weight of every edge in the arborescence. 

Using the fact that there is no duality gap between arborescence packing
and mincut \cite{Edmonds73,Gabow1995}, a standard MWU analysis implies that these arborescences
that we found, after some scaling, form a $(1+\epsilon)$-approximately
optimal fractional arborescence packing. So our arborescence crosses
$S^{*}$ at most $(1+O(\epsilon))<2$ times on average. Thus, if we
sample $O(\log n)$ arborescences from our collections, whp, one of
them will 1-respect $S^{*}$.\footnote{This should be compared with Karger's mincut algorithm in the undirected case, where there is a factor $2$ gap, and hence Karger can only guarantee a $2$-respecting tree in the undirected case.} Below, we formalize this high level
description.

\begin{definition}[Packing problem \cite{Young1995}]
For convex set $P\subseteq \mathds{R}^n$ and nonnegative linear function $f:P\to \mathds{R}^m$, let $\gamma^*=\min_{x\in P}\max_{j\in[m]} f_j(x)$ be the solution in $P$ that minimizes the maximum value of $f_j(x)$ over all $j$, and define the {\em width} of the packing problem as $\omega=\max_{j\in[m],x\in P}f_j(x)-\min_{j\in[m],x\in P}f_j(x)$.
\end{definition}

The fractional arborescence packing problem conforms to this definition.
Enumerate all the $s$-arborescences as $A_1,A_2,\ldots,A_N$. We represent a fractional packing of arborescences as a vector in $\mathds{R}^N$, where coordinate $i$ represents the fractional contribution of $A_i$ in the packing. Let$P=\{x\in \mathds{R}^N : x^T1=1,x\ge0\}$ be the convex hull of all single arborescences.
For each edge $j$ with capacity $w(j)$, $f_j(x)=\sum_{i\in[N]} x_i 1[j\in T_i]/w(j)$ is the relative load of arborescence packing $x$ on edge $j$.
It is easy to see that $\omega\le1/w_{\min}$ for tree packing.
The objective function is to minimize the maximum load: $\gamma^*=\min_{x\in P}\max_{j\in[m]}f_j(x)$. 

For any fractional arborescence packing $x\in\mathds{R}^N$ with value $x^T1=v$ where $f_j(x)\le1$ for all edges $j$, we have $\frac1vx\in P$. In particular, the maximum arborescence packing, once scaled down by its value, is exactly the vector in $P$ that minimizes the maximum load. Therefore, it suffices to look for the vector $x\in P$ achieving the optimal value $\gamma^*$, and then scale the vector up by $1/\gamma^*$ to obtain the maximum arborescence packing.

Next we describe the packing algorithm (Figure 2 of \cite{Young1995}).
Maintain a vector $y\in \mathds{R}^m$, initially set to $y=1$.
In each iteration, find $x=\arg\min_{x\in P}\sum_j y_j f_j(x)$, and then add $x$ to set $S$ and replace $y$ by the vector $y'$ defined by $y'_j=y_j(1+\epsilon f_j(x)/\omega)$.
After a number of iterations, return $\bar{x}\in P$, the average of all the vectors $x$ over the course of the algorithm. The lemma below upper bounds the number of iterations that suffice:

\begin{lemma}[Corollary 6.3 of \cite{Young1995}]
\label{lem:packing}
After $\lceil\frac{(1+\epsilon)\omega\ln m}{\gamma^*((1+\epsilon)\ln(1+\epsilon)-\epsilon)}\rceil$ iterations of the packing algorithm, $\bar{\gamma}=\max_j f_j(\bar{x})\le (1+\epsilon)\gamma^*$.
\end{lemma}

We will also make use of the (exact) duality between $s$-arborescence packing and minimum $s$-cut:

\begin{lemma}[Corollary 2.1 of \cite{Gabow1995}]
The value of maximum $s$-arborescence packing is equal to the value of minimum $s$-cut.
\end{lemma}

\begin{proof}[Proof of \Cref{thm:tree-packing}]
First, construct $H$ according to Theorem \ref{thm:sparsification}.
By the duality above, the minimum $s$-cut on $H$ has value $\lambda_H=\frac{1}{\gamma^*}$.
Since $\min\{|S^*|,|\overline{S^*}|\}\le \theta\sqrt{n}/\log n$, we have $\lambda_H \le \delta_H(S^*)\le(1+\epsilon)p\lambda \le (1+\epsilon)\sqrt{n}$.

Run the aforementioned arborescence packing algorithm up to $O(\lambda_H\ln m)$ iterations, after which Lemma \ref{lem:packing} guarantees that $\bar{\gamma}\le(1+\epsilon)\gamma^*$.
Then $\bar{x}/\bar{\gamma}$ is a vector in $P$ with value $1/\bar{\gamma}\ge \frac1{1+\epsilon}\lambda_H$.

Consider sampling a random arborescence $A$ from the distribution specified by $\bar x/\bar\gamma$, so we choose arborescence $A_i$ with probability $\bar x_i/\bar\gamma$.
Since $\delta_H(S^*) \le (1+\epsilon)p\lambda \le (1+\epsilon)^2\lambda_H$, the expected number of edges in $A\cap\delta_H(S^*)$ is at most $\frac{(1+\epsilon)^2}{1-\epsilon}\le 1+4\epsilon$ for small enough $\epsilon>0$.
Since we always have $|A\cap\delta_H(S^*)|\ge1$, by Markov's inequality $\Pr[|A\cap\delta_H(S^*)|-1\ge 1]\le 4\epsilon \le 1/2$ for small enough $\epsilon$. 
Therefore, if we uniformly sampling $O(\log n)$ arborescences from the distribution $\bar x/\bar\gamma$, at least one of the arborescences is 1-respecting whp.

It remains to compute $x=\arg\min_{x\in P}\sum_j y_j f_j(x)$ on each iteration. Since $\sum_jy_jf_j(x)$ is linear in $x$, the  minimum must be achieved by a single arborescence.
So the task reduces to computing the minimum cost spanning $s$-arborescence, which can be done in $O(m+n\log n)$ \cite{Gabow1986}. The total time complexity, over all iterations, becomes $O((m+n\log n)\lambda_H\log n) = O((m+n\log n)\sqrt{n}\log n)$.
\end{proof}

\section{Mincut Given 1-respecting Arborescence}
\label{sec:maxflows}

We propose an algorithm (Algorithm~\ref{alg:tree}) that uses $O(\log n)$ maxflow subroutines to find the minimum $s$-cut  that $1$-respects a given $s$-arborescence.
The result is formally stated in Theorem \ref{thm:tree-alg}.

\begin{theorem} \label{thm:tree-alg}
Consider a directed graph $G = (V, E, w)$ with $n$ vertices, $m$ edges, and polynomially bounded edge weights $w_e > 0$.
Fix a global (directed) mincut $S$ of $G$.
Given an arborescence $T$ rooted at $s \in S$ with $|T \cap (S, \overline{S})| = 1$, Algorithm~\ref{alg:tree} outputs a global minimum cut of $G$ in time $O((MF(m,n) + m)\cdot \log n)$.
\end{theorem}

We first give some intuition for Algorithm~\ref{alg:tree}.
Because $s \in S$, if we could find a vertex $t \in \overline{S}$, then computing the $s$-$t$ mincut using one maxflow call would yield a global mincut of $G$.  However, we cannot afford to run one maxflow between $s$ and every other vertex in $G$.
Instead, we carefully partition the vertices into $\ell = O(\log n)$ sets $(C_i)_{i=1}^{\ell}$.
We show that for each $C_i$, we can modify the graph appropriately so that it allows us to (roughly speaking) compute the maximum flow between $s$ and every vertex $c \in C_i$ using one maxflow call.

More specifically, Algorithm~\ref{alg:tree} has two stages. In the first stage, we compute a \emph{centroid decomposition} of $T$. Recall that a centroid of $T$ is a vertex whose removal disconnects $T$ into subtrees with at most $n/2$ vertices. This process is done recursively, starting with the root $s$ of $T$. We let $P_1$ denote the subtrees resulting from the removal of $s$ from $T$. In each subsequent step $i$, we compute the set $C_i$ of the centroids of the subtrees in $P_i$. We then remove the centroids and add the resulting subtrees to $P_{i+1}$. This process continues until no vertices remain.

In the second stage, for each layer $i$, we construct a directed graph $G_i$ and perform one maxflow computation on $G_i$. The maxflow computation on $G_i$ would yield candidate cuts for every vertex in $C_i$, and after computing the appropriate maximum flow across every layer, we output the minimum candidate cut as the minimum cut of $G$. The details are presented in Algorithm~\ref{alg:tree}.

\begin{algorithm}[!ht]
\caption{Finding the global minimum directed cut.}
\label{alg:tree}
\SetKwInOut{Input}{Input}
\SetKwInOut{Output}{Output}
\setcounter{AlgoLine}{0}
\Input{An arborescence $T$ rooted at $s\in S$ such that $S$ 1-respects $T$.}
\smallskip
// Stage I: Build centroid decomposition. \\
\smallskip
Let $C_0 = \{s\}$, $P_1 =$ the set of subtrees obtained by removing $s$ from $T$, and $i = 1$. \\
\While{$P_i\ne\varnothing$}{
Initialize $C_i$ (the centroids of $P_i$) and $P_{i+1}$ as empty sets. \\
\For{each subtree $U\in P_i$}{
Compute the centroid $u$ of $U$ and add it to $C_i$. \\
    Add all subtrees generated by removing $u$ from $U$ to $P_{i+1}$.
    }
Set $\ell = i$ and iterate $i = i+1$.
}
\smallskip
// Stage II: Calculate integrated maximum flow for each layer. \\
\smallskip
\For{$i = 1$ {\bf to} $\ell$}{
Construct a digraph $G_i = (V \cup \{t_i\}, E_1 \cup E_2 \cup E_3)$ as follows (see \Cref{fig:my_label}): \label{alg:step11}\\
    \quad 1) Add edges $E_1 = E \cap \cup_{U \in P_i} (U\times U)$ with capacity equal to their original weight.\label{alg:step12}\\
    \quad 2) Add edges $E_2 = \{(s,v):(u,v)\in E\setminus E_1\}$ with capacity of $(s,v)$ equal to the original weight of $(u,v)$. \\ %
    \quad 3) Add edges $E_3 = \{(u,t_i):u\in C_i\}$ with infinite capacity.\label{alg:step14}\\
Compute the maximum $s$-$t_i$ flow $f^*_i$ in $G_i$. \label{alg:step15} \\
For each component $U\in P_i$ with centroid $u$, the value of $f^*_i$ on edge $(u,t_i)$ is a candidate cut value, and the nodes in $U$ that can reach $u$ in the residue graph is a candidate for $\overline{S}$.
}
Return the smallest candidate cut as minimum $s$-cut and the corresponding $(S, \overline{S})$.
\end{algorithm}

\begin{figure}
    \centering
    \includegraphics[width=\linewidth]{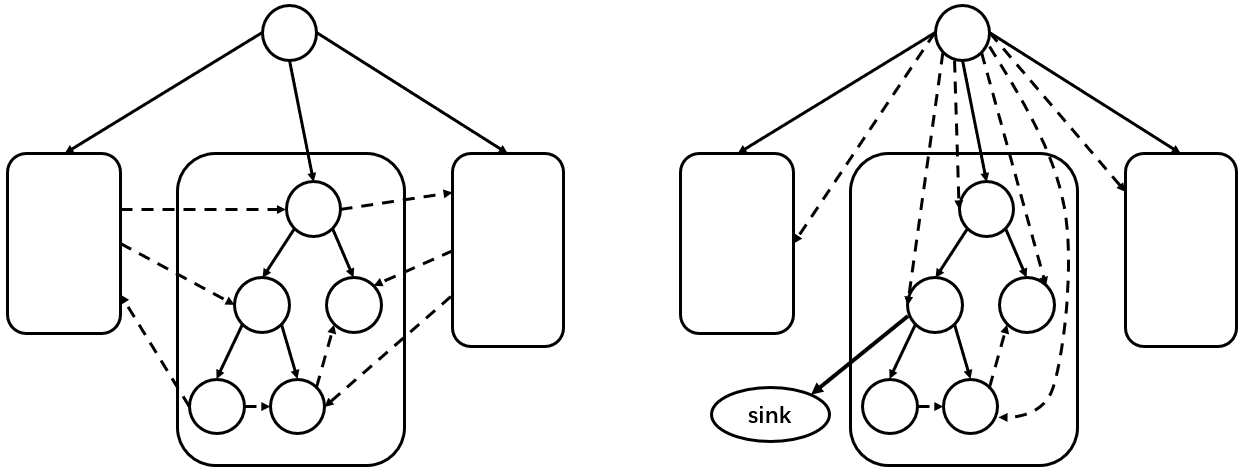}
    \caption{
    Construction of auxiliary graph $G_i$ in Algorithm \ref{alg:tree}.
    Solid lines represent the arborescence $T$.
    Dashed lines are other edges in the graph.
    Rectangles are sets formed by the first level of centroid decomposition.
    Left: The original graph.
    Right: The part of $G_1$ solving the case that the mincut separates root and the centroid of the middle subtree.%
    }
    \label{fig:my_label}
\end{figure}

We first state two technical lemmas that we will use to prove Theorem~\ref{thm:tree-alg}.

\begin{lemma} \label{lem:tree-barS-inone}
Recall that $P_i$ is the set of subtrees in layer $i$ and $C_i$ contains the centroid of each subtree in $P_i$.
If $C_{j} \subseteq S$ for every $0 \le j < i$, then $\overline{S}$ is contained in exactly one subtree in $P_i$, and consequently, at most one vertex $u \in C_i$ can be in $\overline{S}$.
\end{lemma}

\begin{lemma} \label{lem:tree-fi}
Let $G_i$ be the graph constructed in Step~\ref{alg:step11} of Algorithm~\ref{alg:tree}.
Let $f^*_i$ be a maximum $s$-$t_i$ flow on $G_i$ as in Step~\ref{alg:step15}.
For any $U \in P_i$ with centroid $u$, the amount of flow $f^*_i$ puts on edge $(u, t_i)$ is equal to the value of the minimum cut between $\overline{U}$ and $u$.
\end{lemma}

We defer the proofs of Lemmas~\ref{lem:tree-barS-inone}~and~\ref{lem:tree-fi}, and first use them to prove Theorem~\ref{thm:tree-alg}.

\begin{proof}[Proof of Theorem~\ref{thm:tree-alg}]
We first prove the correctness of Algorithm~\ref{alg:tree}.

Because $C_0 = \{s\}$ and $s \in S$, and the $C_i$'s form a disjoint partition of $V$, there must be a layer $i$ such that for the first time, we have a centroid $u \in C_i$ that belongs to $\overline{S}$.
By Lemma~\ref{lem:tree-barS-inone}, we know that $\overline{S}$ must be contained in exactly one subtree $U \in P_i$, and hence $u$ must be the centroid of $U$.
In summary, we have $u \in \overline{S}$ and $\overline{S} \subseteq U$.

Consider the graph $G_i$ constructed for layer $i$.
By Lemma~\ref{lem:tree-fi}, based on the flow $f^*_i$ puts on the edge $(u, t_i)$, we can recover the value of the minimum (directed) cut between $\overline{U}$ and $u$.
Because $\overline{S} \subseteq U$ (or equivalently $\overline{U} \subseteq S$) and $u \in \overline{S}$, the cut $(S, \overline{S})$ is one possible cut that separates $\overline{U}$ and $u$. Therefore, the flow that $f^*_i$ puts on the edge $(u, t_i)$ is equal to the global mincut value in $G$.

In addition, the candidate cut value for any other centroid $u'$ of a subtree $U' \in P_i$ must be at least the mincut value between $s$ and $u'$.
This is because the additional restriction that the cut has to separate $\overline{U'}$ from $u'$ can only make the mincut value larger, and the value of this cut in $G_i$ is equal to the value of the same cut in $G$.
Therefore, the minimum candidate cut value in all $\ell$ layers must be equal to the global mincut value of $G$.

Now we analyze the running time of Algorithm~\ref{alg:tree}. We can find the centroid of an $n$-node tree in time $O(n)$ (see e.g., \cite{Megiddo1981}).
The total number of layers $\ell = O(\log n)$ because removing the centroids reduces the size of the subtrees by at least a factor of $2$. Thus, the running time of Stage I of Algorithm~\ref{alg:tree} is $O(n \log n)$. In Stage II, we can construct each $G_i$ in $O(m)$ time and every $G_i$ has $O(m)$ edges.
Since there are $O(\log n)$ layers and the maximum flow computations take a total of $O(MF(m,n)\cdot \log n)$ time,
the overall runtime is $O(n \log n + (MF(m, n) + m) \log n) = O((MF(m, n) + m) \log n)$.
\end{proof}

Before proving Lemmas~\ref{lem:tree-barS-inone}~and~\ref{lem:tree-fi} we first prove the following lemma.

\begin{lemma} \label{lem:tree-barS-path}
If $x$ and $y$ are vertices in $\overline{S}$, then every vertex on the (undirected) path from $x$ to $y$ in the arborescence $T$ also belongs to $\overline{S}$.
\end{lemma}

\begin{proof}
Consider the lowest common ancestor $z$ of $x$ and $y$.
Because there is a directed path from $z$ to $x$ and a directed path from $z$ to $y$, we must have $z \in \overline{S}$.
Otherwise, there are at least two edges in $T$ that go from $S$ to $\overline{S}$.

Because $s \in S$ and $z \in \overline{S}$, there is already an edge in $T$ (on the path from $s$ to $z$) that goes from $S$ to $\overline{S}$.
Consequently, all other edges in $T$ cannot go from $S$ to $\overline{S}$, which means the entire path from $z$ to $x$ (and similarly $z$ to $y$) must be in $\overline{S}$.
\end{proof}

Recall that Lemma~\ref{lem:tree-barS-inone} states that if all the centroids in previous layers are in $S$, then $\overline{S}$ is contained in exactly one subtree $U$ in the current layer $i$.

\begin{proof}[Proof of Lemma~\ref{lem:tree-barS-inone}]
For contradiction, suppose that there exist distinct subtrees $U_1$ and $U_2$ in $P_i$ and vertices $x,y \in \overline{S}$ such that $x \in U_1$ and $y \in U_2$.

By Lemma~\ref{lem:tree-barS-path}, any vertex on the (undirected) path from $x$ to $y$ also belongs to $\overline{S}$.
Consider the first time that $x$ and $y$ are separated into different subtrees.
This must have happened because some vertex on the path from $x$ to $y$ is removed. However, the set of vertices removed at this point of the algorithm is precisely $\bigcup_{0 \le j < i} C_j$, but our hypothesis assumes that none of them are in $\overline{S}$.
This leads to a contradiction and therefore $\overline{S}$ is contained in exactly one subtree of $P_i$.

It follows immediately that at most one centroid $u \in C_i$ can be in $\overline{S}$.
\end{proof}

Next we prove Lemma~\ref{lem:tree-fi}, which states that the maximum flow between $s$ and $t_i$ in the modified graph $G_i$ allows one to simultaneously compute a candidate mincut value for each vertex $u \in C_i$.

\begin{proof}[Proof of Lemma~\ref{lem:tree-fi}]
First observe that the maxflow computation from $s$ to $t_i$ in $G_i$ can be viewed as multiple independent maxflow computations.
The reason is that, for any two subtrees $U_1, U_2 \in P_i$, there are only edges that go from $s$ into $U_1$ and from $U_1$ to $t_i$ in $G_i$ (similarly for $U_2$), but there are no edges that go between $U_1$ and $U_2$.

The above observation allows us to focus on one subtree $U \in P_i$.
Consider the procedure that we produce $G_i$ from $G$ in Steps~\ref{alg:step12}~to~\ref{alg:step14} of Algorithm~\ref{alg:tree}.
The edges with both ends in $U$ are intact (the edge set $E_1$).
If we contract all vertices out of $U$ into $s$, then all edges that enter $U$ would start from $s$, which is precisely the effect of removing cross-subtree edges and adding the edges in $E_2$.
One final infinity-capacity edge $(u, t_i) \in E_3$ connects the centroid of $U$ to the super sink $t_i$.

Therefore, the maximum $s$-$t_i$ flow $f^*_i$ computes the maximum flow between $\overline{U}$ and $u \in U$ simultaneously for all $U \in P_i$, whose value is reflected on the edge $(u, t_i)$.
It follows from the maxflow mincut theorem that the flow on edge $(u, t_i)$ is equal to the mincut value between $\overline{U}$ and $u$ in $G$ (i.e., the minimum value $w(A, \overline{A})$ among all $A \subset V$ with $\overline{U} \subseteq A$ and $u \in \overline{A}$).
\end{proof}

\bibliographystyle{alpha}
\bibliography{refs}

\end{document}